\documentclass[twocolumn,aps,pra,showpacs,superscriptaddress]{revtex4-1}

\usepackage[utf8]{inputenc}
\usepackage[T1]{fontenc}
\usepackage[USenglish]{babel}
\usepackage[sc,osf]{mathpazo}
\usepackage[scaled=0.86]{berasans}
\usepackage[scaled=1.03]{inconsolata}
\usepackage[colorlinks]{hyperref}
\usepackage{graphicx}
\usepackage{amsmath,amssymb,amsthm}
\usepackage{xspace}

\let\originalleft\left
\let\originalright\right
\renewcommand{\left}{\mathopen{}\mathclose\bgroup\originalleft}
\renewcommand{\right}{\aftergroup\egroup\originalright}

\newcommand{\de}[1]{\left(#1\right)}

\newcommand{\mean}[1]{\langle#1\rangle}
\newcommand{\ket}[1]{\left| #1 \right\rangle}

\newcommand{\ketbra}[2]{\left|#1\middle\rangle\middle\langle#2\right|}

\newcommand{\eg}{\textit{e.g.}\@\xspace}
\newcommand{\ie}{\textit{i.e.}\@\xspace}

\newcommand{\id}{\mathbb{1}}
\DeclareMathOperator{\tr}{tr}

\newtheorem{theorem}{Theorem}
\newtheorem{lemma}{Lemma}

\hypersetup{pdftitle={{All noncontextuality inequalities for the n-cycle scenario}}}

\begin{document}

\title{\texorpdfstring{All noncontextuality inequalities for the $n$-cycle scenario}{All noncontextuality inequalities for the n-cycle scenario}}

\author{Mateus Araújo}
 %\email{maltusan@gmail.com}
 \thanks{These authors contributed equally to this work.}
 \affiliation{Departamento de Física, Universidade Federal de Minas Gerais,
 Caixa Postal 702, 30123-970 Belo Horizonte, MG, Brazil}
 \affiliation{Faculty of Physics, University of Vienna, Boltzmanngasse 5, 1090 Vienna, Austria}

\author{Marco Túlio Quintino}
 %\email{mtcq.mm@gmail.com}
 \thanks{These authors contributed equally to this work.}
 \affiliation{Departamento de Física, Universidade Federal de Minas Gerais,
 Caixa Postal 702, 30123-970 Belo Horizonte, MG, Brazil}
 \affiliation{Département de Physique Théorique, Université de Genève, 1211 Genève, Switzerland}

\author{Costantino Budroni}
 %\email{cbudroni@us.es}
  \affiliation{Naturwissenschaftlich-Technische Fakult\"at, Universit{\"a}t Siegen,
Walter-Flex-Straße 3, D-57068 Siegen,Germany}
 \affiliation{Departamento de F\'{\i}sica Aplicada II, Universidad de Sevilla, E-41012 Sevilla, Spain}

\author{Marcelo Terra Cunha}
 %\email{tcunha@mat.ufmg.br}
 \affiliation{Departamento de Matem\'atica, Universidade Federal de Minas Gerais,
 Caixa Postal 702, 30123-970 Belo Horizonte, MG, Brazil}

\author{Ad\'an Cabello}
 %\email{adan@us.es}
  \affiliation{Departamento de Física, Universidade Federal de Minas Gerais, Caixa Postal 702, 30123-970 Belo Horizonte, MG, Brazil}
 \affiliation{Departamento de F\'{\i}sica Aplicada II, Universidad de Sevilla, E-41012 Sevilla, Spain}

%%%%%%%%%%%%%%%%%%%%%%%%%%%%%%%%%%%%%%%%%%%%%%%%%%%%%%%%%%%%%%%%%%%

\date{\today}

%First version: April 30, 2012 (Belo Horizonte)
%Second version: May 1, 2012 (Sevilla)
%nth version: May 23, 2012 (Singapore)
%This version: Jul 26, 2012 (Belo Horizonte)
%(n+2)th version: 02.12.2012 (Vienna)
%(n+3)th version: 25.01.2013 (Beijing)
%(n+4)th version: 05.02.2013 (Wien)
%(n+5)th version: 03.04.2013 (Wien)

%%%%%%%%%%%%%%%%%%%%%%%%%%%%%%%%%%%%%%%%%%%%%%%%%%%%%%%%%%%%%%%%%%%

\begin{abstract}
The problem of separating classical from quantum correlations is in general intractable and has been solved explicitly only in few cases. In particular, known methods cannot provide general solutions for an arbitrary number of settings. We provide the complete characterization of the classical correlations and the corresponding maximal quantum violations for the case of $n\geq 4$ observables $X_0,\ldots,X_{n-1}$, where each consecutive pair $\{X_i,X_{i+1}\}$, sum modulo $n$, is jointly measurable. This generalizes both the Clauser-Horne-Shimony-Holt and the Klyachko-Can-Binicio\u{g}lu-Shumovsky scenarios, which are the simplest ones for, respectively, locality and noncontextuality. In addition, we provide explicit quantum states and settings with maximal quantum violation and minimal quantum dimension.
\end{abstract}

%%%%%%%%%%%%%%%%%%%%%%%%%%%%%%%%%%%%%%%%%%%%%%%%%%%%%%%%%%%%%%%%%%%

%Foundations of quantum mechanics; measurement theory
%Entanglement and quantum nonlocality
%(e.g. EPR paradox, Bell's inequalities, GHZ states, etc.)

\maketitle

%%%%%%%%%%%%%%%%%%%%%%%%%%%%%%%%%%%%%%%%%%%%%%%%%%%%%%%%%%%%%%%%%%%
% Introduction
%%%%%%%%%%%%%%%%%%%%%%%%%%%%%%%%%%%%%%%%%%%%%%%%%%%%%%%%%%%%%%%%%%%

\section{Introduction}
Quantum correlations among the results of jointly measurable observables go beyond the limits of classical correlations and provide a whole new set of resources for physics \cite{rio11}, computation \cite{janet09}, and communication \cite{gisin02,buhrman10}. Yet, surprisingly, necessary and sufficient conditions for classicality -- all the Bell/noncontextuality inequalities -- are known only for a few scenarios, the most famous being the Clauser-Horne-Shimony-Holt (CHSH) scenario \cite{chsh69}, completely characterized in \cite{fine82}, and the Klyachko-Can-Binicio\u{g}lu-Shumovsky (KCBS) scenario \cite{klyachko02,klyachko08}. In both cases, quantum correlations go beyond the classical ones \cite{bell64,klyachko08}.

	Unlike the CHSH scenario, the KCBS scenario cannot be associated with correlations among the results of measurements on different subsystems, but rather to the results of measurements on a single system \cite{cabello08,kirchmair09,amselem09,lapkiewicz11}. In this case, the existence of quantum correlations outside the classical set shows the impossibility of noncontextual hidden variable (NCHV) theories \cite{specker60,seevinck11,bell66,kochen67}. Quantum contextuality is a natural generalization of quantum nonlocality that neither privileges spacelike-separated observables (among other jointly measurable observables), nor composite systems (among other physical systems), nor entangled states (among other quantum states), and provides advantage versus classical (noncontextual) resources even in scenarios with no spacelike separation \cite{cabello11b,nagali12,amselem12}.

	Both the CHSH and KCBS scenarios can be understood as particular cases of a much larger family: The scenario of $n$ dichotomic observables $X_i$ such that the pairs $\{X_i,X_{i+1}\}$, modulo $n$, are jointly measurable. If we represent observables as nodes of a graph and link them with edges when they are jointly measurable, the resulting graph is the $n$-cycle (see Fig. \ref{Fig0}). Besides the CHSH and KCBS scenarios, \ie, $n=4,5$, also other cases have been completely characterized: the cases $n=2$ \cite{boole62,pitowsky94} and $n=3$ \cite{specker60,pitowsky89,liang11}, and a partial characterization have been given in \cite{liang11,cabello10} (for odd $n$) and in \cite{fritz11,chaves12,chaves13} (in terms of entropic inequalities, necessary but not sufficient conditions, for any $n$) and the case $n=6$ has been discussed in \cite{lapkiewicz11} in relation with the test of the KCBS inequality.

	The main difficulty in the characterization of correlations resides in the fact that the existing general approaches for obtaining classical \cite{pitowsky89} and quantum \citep{navascues07,navascues08} bounds involve the use of algorithms that must be applied to specific cases and that require an amount of resources for computation rapidly growing with the number of settings. In fact, the only known case in which a complete characterization of classical bounds and the corresponding quantum violation has been given for any number $n$ of settings is the bipartite Bell scenario in which Alice can choose between two dichotomic observables and Bob among $n$ \cite{sliwa03,collins04}.
	
	In this paper we provide the complete set of noncontextuality inequalities, \ie, necessary and sufficient conditions for noncontextuality, and the corresponding quantum violations for the $n$-cycle. Moreover, we exhibit quantum states and measurements which maximally violate the noncontextuality inequalities for each $n$ with minimum dimension of the corresponding Hilbert space.

%%%%%%%%%%%%%%%%%%%%%%%%%%%%%%%%%%%%%%%%%%%%%%%%%%%%%%%%%%%%%%%%%%%
% Fig. 0
%%%%%%%%%%%%%%%%%%%%%%%%%%%%%%%%%%%%%%%%%%%%%%%%%%%%%%%%%%%%%%%%%%%

\begin{figure}[t]
\centerline{\includegraphics[width=8.4cm]{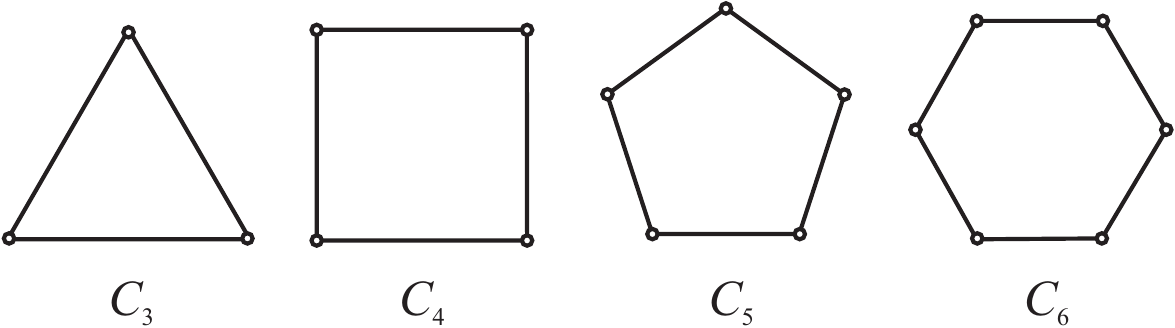}}
\caption{\label{Fig0}Graphs associated to the compatibility relations among the observables $X_i$ for $n=3,4,5,6$. $C_4$ corresponds to CHSH case with the labelling of nodes $A_1,B_1,A_2,B_2$, in the usual notation for Alice and Bob observables, and $C_5$ corresponds to KCBS case with the labelling $X_0,\ldots, X_4$.}
\end{figure}

%%%%%%%%%%%%%%%%%%%%%%%%%%%%%%%%%%%%%%%%%%%%%%%%%%%%%%%%%%%%%%%%%%%

\section{Preliminary notions} The simplest way to introduce the notion of noncontextuality is by analogy with the well-known notion of locality (\eg, \cite{abramsky11,chaves12,kleinmann12}). Here we shall follow \cite{chaves12}. The difference between the two resides in the definition of joint measurability: One no longer considers only joint measurements of spacelike-separated observables, but also admits the joint measurement of a collection of mutually compatible observables -- a \textit{context}. Consequently, the assumption of context-independence for outcomes replace the assumption of locality, \ie, independence between spacelike-separated measurements. We recall that an operational definition of compatibility can be given independently of quantum formalism, \ie, without referring to commutativity \cite{peres93,guhne10}.

	More precisely, given a set of observables $\{X_0,\ldots, X_{n-1}\}$, a context $\mathfrak{c}$ is a set of indices such that $X_i$ is compatible with $X_j$ whenever $i,j\in \mathfrak{c}$. Notice that all subsets of $\mathfrak{c}$, including one-element sets, must be admissible contexts. A contextuality scenario is therefore given by a set of observables $\{X_i\}$ together with the set of admissible contexts $\mathcal{C}=\{ \mathfrak{c}_k\}$, or simply the maximal ones. For each context one shall, then, measure joint statistics for its observables, and the set of all these statistics for some given contextuality scenario is known simply as correlations.

	Analogously with the study of locality, we shall consider here three kinds of correlations: no-disturbance, quantum, and noncontextual. The no-disturbance condition here is a simple generalization of the well-known no-signalling condition, that applies not only to observables that act on separate subsystems, but to any set of observables that are in a context \cite{ramanathan12}. As the set of no-signalling correlations, the set of correlations that respect no-disturbance is also a polytope, the \textit{no-disturbance} polytope.
	
	If our correlations respect no-disturbance and come from dichotomic observables (as will always be the case in this paper), then they can be always represented as a vector $\mathbf{v}=(v_{\mathfrak{c}} | \mathfrak{c}\in\mathcal{C})$, where $v_\mathfrak{c}$ is the expectation value of the product of the observables in context $\mathfrak{c}$ \cite{budroni10}. Deterministic noncontextual classical models assign a definite outcome $x=(x_0,\ldots,x_{n-1})\in\{-1,1\}^n$ to the observables $X_0,\ldots,X_{n-1}$, and the assignments for the correlations within each context are thus given by $v_\mathfrak{c}= \prod_{i\in \mathfrak{c}} x_i$, in a context-independent way.

	As a consequence, the set of correlations consistent with a noncontextual model is given by the convex hull of the deterministic assignments for the correlation vector $\mathbf{v}$ -- the noncontextual polytope. Tight noncontextuality inequalities are therefore affine bounds defined as the facets of the noncontextuality polytope, namely $(p-1)$-dimensional faces of a $p$-dimensional polytope. In this sense, tight inequalities are the minimal set of necessary and sufficient conditions for classicality of correlations.

	The $n$-cycle contextuality scenario is given by $n$ observables $X_0,\ldots,X_{n-1}$ and the set of maximal contexts\begin{equation}\mathcal{C}_{n} = \{\{X_0,X_1\},\ldots,\{X_{n-2},X_{n-1}\},\{X_{n-1},X_0\}\}. 
	\end{equation}
	It can be depicted as a graph where nodes represent observables and edges represent joint measurability (see Fig. \ref{Fig0}). All correlations are then given by the $2n$-dimensional vector
	\begin{equation}\label{eq:corrvec}
	(\langle X_0\rangle,\ldots,\langle X_{n-1}\rangle,\langle X_0 X_1\rangle,\ldots,\langle X_{n-1}X_0\rangle).
	\end{equation}

	The no-disturbance polytope for this scenario is easy to characterize. Since representing the correlations via expectation values already implies no-disturbance and normalization of the probabilities, the only condition left to enforce is their positivity. This condition, written in terms of elements of the vector \eqref{eq:corrvec}, gives us\small
	\begin{subequations}\label{eq:positivity}	
	\begin{align} 
	\label{eq:um}
	4p(++|X_iX_{i+1}) &= 1 + \mean{X_i} + \mean{X_{i+1}} + \mean{X_iX_{i+1}} \ge 0 \\
	\label{eq:dois}
	4p(+-|X_iX_{i+1}) &= 1 + \mean{X_i} - \mean{X_{i+1}} - \mean{X_iX_{i+1}} \ge 0 \\
	\label{eq:tres}
	4p(-+|X_iX_{i+1}) &= 1 - \mean{X_i} + \mean{X_{i+1}} - \mean{X_iX_{i+1}} \ge 0 \\
	\label{eq:quatro}
	4p(--|X_iX_{i+1}) &= 1 - \mean{X_i} - \mean{X_{i+1}} + \mean{X_iX_{i+1}} \ge 0,
	\end{align}
	\end{subequations}
	\normalsize
	which are the facets of the no-disturbance polytope. 
%which are simply the positivity conditions for the terms appearing in (\ref{eq:nodist}) written in terms of the terms appearing in (). %These are the only ones required, since normalization and no-disturbance are already implied by representing the correlations as expectation values.

\section{Main result}In the remainder of the paper, we shall always take $n \ge 3$. For the $2$-cycle the only facets of the noncontextual polytope are the four positivity conditions \eqref{eq:positivity}, \ie, the noncontextual polytope coincides with the no-disturbance polytope.

\begin{theorem}\label{th:main}
All $2^{n-1}$ tight noncontextuality inequalities for the $n$-cycle noncontextual polytope are
\begin{equation}
 \label{eq:boolenciclo}
 \Omega=\sum_{i=0}^{n-1} \gamma_i \langle X_{i} X_{i+1} \rangle
\stackrel{\text{\tiny{\textup{NCHV}}}}{\leq} n - 2,
\end{equation}
where $\gamma_i \in \{-1,1\}$ such that the number of $\gamma_i = -1$ is odd.\end{theorem}
\begin{proof}
We apply the method based on the results of \cite{budroni10} and presented in \cite{budroni12a}, namely that the existence of a classical model for a set of observables is equivalent to the existence of classical models for particular subsets coinciding on their intersection. In our proof, we use that the existence of a classical probability model for the observables  $\{X_0,\ldots,X_{n-1}\}$ is equivalent to the existence of classical models for $\{X_0,\ldots,X_{n-2}\}$ and $\{X_0,X_{n-1},X_{n-2}\}$, coinciding on their intersection $\{X_0,X_{n-2}\}$ (see Fig. \ref{Fig2}). Such a consistency condition for the intersection is written in terms of the ``unmeasurable correlation'' $\langle X_0 X_{n-2}\rangle$, \ie, a correlation between observables that are not in a context and therefore cannot be jointly measured, but have nevertheless a well-defined correlation in every classical model \cite{budroni12b}. The final set of inequalities must not contain the variable $\langle X_0 X_{n-2}\rangle$, which must be removed by applying Fourier-Motzkin (FM) elimination \cite{ziegler95}, \ie, by summing inequalities where it appears with the minus sign, with those where it appears with the plus sign. This step of the proof is a simple application of the techniques from \cite{budroni12a}. For the convenience of the reader, details are presented in Appendix \ref{app:a}.

We can now proceed by induction on $n$. The case $n=3$ is known. For the inductive step, following the above argument, we calculate the $n$-cycle inequalities by combining the $(n-1)$-cycle inequalities for the subset $\{X_0,\ldots,X_{n-2}\}$ with the $3$-cycle inequalities for $\{X_0,X_{n-1},X_{n-2}\}$. We apply FM elimination on the variable $\langle X_0 X_{n-2}\rangle$ from the whole set of inequalities. All inequalities in (\ref{eq:boolenciclo}) are obtained by combining one inequality for the $(n-1)$-cycle with one for the $3$-cycle, and are in the right number. Combining two inequalities for the $(n-1)$-cycle, or two for the $3$-cycle gives a redundant inequality, as happens for combination of positivity conditions (\ref{eq:positivity}) with inequalities of the form (\ref{eq:boolenciclo}), the latter being obtainable as a sum of $n-1$ (or $3$) positivity conditions. There are no other inequalities. The proof of their tightness is presented in Appendix \ref{app:tightness}.
\end{proof}

The reader familiar with Fine's proof for the $4$-cycle \cite{fine82}, obtained by combining two $3$-cycles, may have noticed that the above is a straightforward generalization.

We can also characterize the vertices of the no-disturbance polytope.

\begin{theorem}\label{teo:nd}
The vertices of the no-disturbance polytope are the $2^n$ noncontextual deterministic correlation vectors
\begin{equation}\label{eq:nc}
(\langle X_0\rangle,\ldots,\langle X_{n-1}\rangle,\langle
X_0\rangle\langle X_1\rangle,\ldots,\langle X_{n-1}\rangle\langle
X_0\rangle),
\end{equation}
where $\langle X_i\rangle = \pm 1$, together with the $2^{n-1}$ contextual correlation vectors of the form
\begin{equation}\label{eq:c}
(0,\ldots,0,\langle X_0 X_1\rangle,\ldots,\langle X_{n-1} X_0\rangle),
\end{equation}
where $\langle X_i X_{i+1}\rangle=\pm 1$ such that number of negative components is odd.
\end{theorem}
\begin{proof}
By definition, the vertices of the polytope are given by the intersection of $2n$ independent hyperplanes, \ie, as a unique solution for a set of $2n$ independent linear equations chosen among the $4n$ equations saturating \eqref{eq:positivity}. The above vertices are obtained
 by choosing two equations among \eqref{eq:um}-\eqref{eq:quatro}, for each index $i$. In particular,  contextual vertices are obtained by choosing equations \eqref{eq:um} and \eqref{eq:quatro} for an odd number of indexes $i$ and equations \eqref{eq:dois} and \eqref{eq:tres} for the remaining indexes. It is straightforward to check that all other possible strategies for obtaining a vertex, \ie, involving the choice of one, two or three equations for each index $i$, give the same set of vertices.\end{proof}

\begin{figure}[t]
	\centering
	\includegraphics[width=0.75\columnwidth]{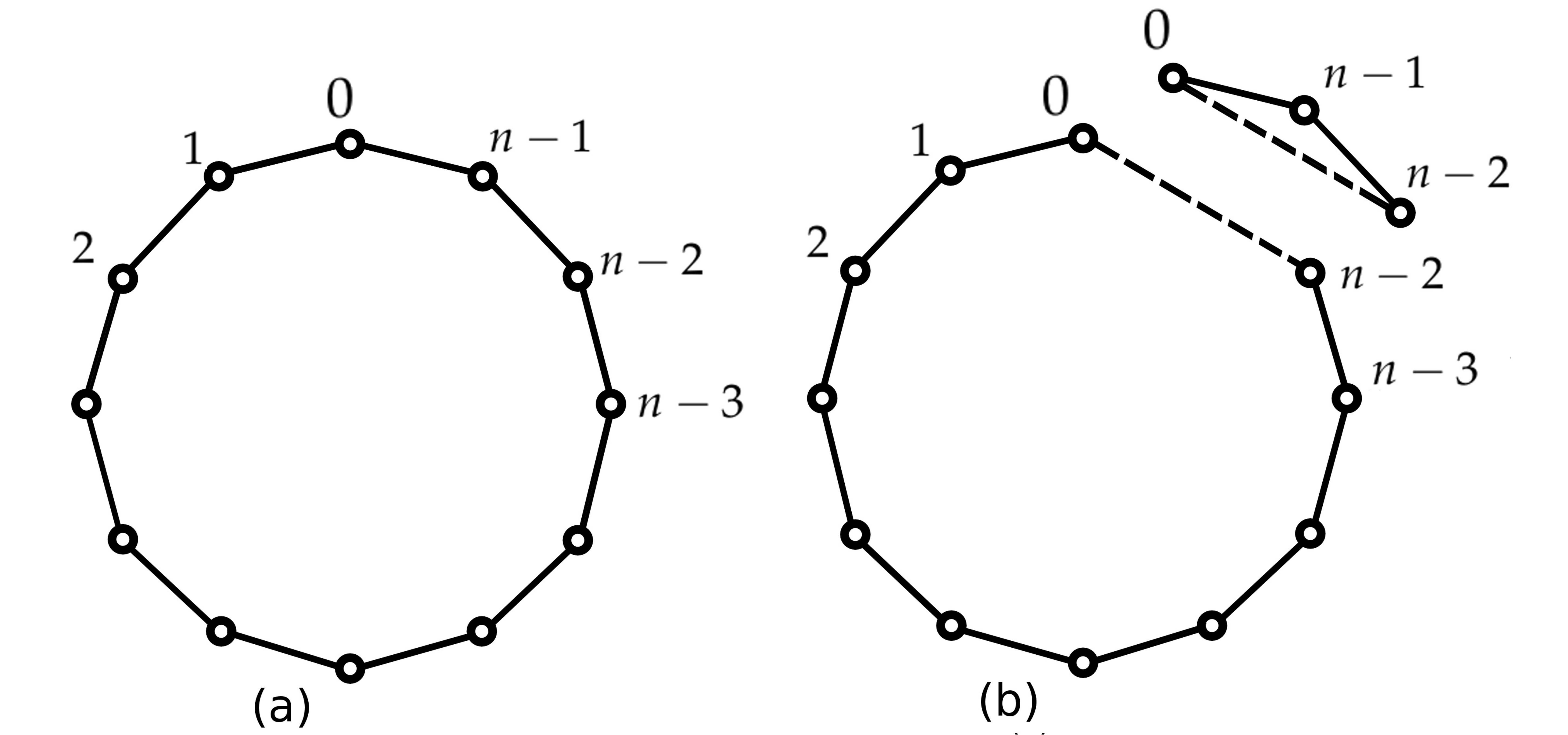}
	\caption{\label{Fig2} (a) $n$-cycle scenario. (b) Subsets of observables that can be associated with the $(n-1)$-cycle and $3$-cycle scenario by considering the ``unmeasurable correlation'' $\mean{X_0 X_{n-2}}$ (dashed line). }
	\label{fig:3cycle}
\end{figure}

	To summarize our results: The no-disturbance polytope, defined by the $4n$ positivity conditions \eqref{eq:positivity}, has $2^n + 2^{n-1}$ vertices, of which $2^n$ are noncontextual and $2^{n-1}$ are contextual. The noncontextuality polytope, defined by the $2^n$ noncontextual vertices (\ref{eq:nc}), has $4n + 2^{n-1}$ facets (it is trivial to check that inequalities (\ref{eq:positivity}) are tight for the noncontextuality polytope). Also note that for each vertex in (\ref{eq:c}) there exists an inequality in (\ref{eq:boolenciclo}) such that $\langle X_i X_{i+1}\rangle=\gamma_i$, \ie, contextual vertices and noncontextuality inequalities are in a one-to-one correspondence. 

\section{Quantum violations}Here we address the problems of whether quantum mechanics violates the inequalities \eqref{eq:boolenciclo}, which is the maximum quantum violation -- the Tsirelson bound --, and how to achieve it.

\begin{theorem}
Quantum mechanics violates the noncontextuality inequalities \eqref{eq:boolenciclo} for any $n \ge 4$. The Tsirelson bound is
\begin{equation}
 \Omega_{\rm QM}=
 \begin{cases}
 \frac{3 n \cos \left(\frac{\pi}{n}\right)-n}{1+\cos \left(\frac{\pi}{n}\right)}& \text{ for odd }n,\\
 n \cos \left(\frac{\pi}{n}\right) & \text{ for even } n.
 \end{cases}
\end{equation}
\end{theorem}

\begin{proof}
Without loss of generality, we can restrict our discussion to the inequalities in which, for odd $n$, $\gamma_i=-1$
for all $i$ and, for even $n$, $\gamma_i=-1$ for all $i$ except
$\gamma_{n-1}=1$. Using that
\begin{multline}
 \pm \langle X_i X_{i+1}\rangle = 2 \Big[p(+\pm|X_i,X_{i+1}) \\ +p(-\mp|X_i,X_{i+1})\Big]-1,
\end{multline}
we can rewrite $\Omega$ as $2 \Sigma - n$, where $\Sigma$ is a sum of probabilities.

Any sum of probabilities is upperbounded in quantum mechanics by the Lovász $\vartheta$-function, $\vartheta(G)$, of the graph $G$ in which nodes are the arguments of the probabilities and edges link exclusive events (\eg, ${(++|X_0,X_1)}$ and $(--|X_1,X_2)$) \cite{cabello10}.

If $n$ is odd, the graph $G$ associated to $\Sigma$ is the prism graph of order $n$, $Y_n$ (see Fig.~\ref{Fig1}). Its $\vartheta$-function is
\begin{equation}\label{eq:thetay}
 \vartheta(Y_n)= \frac{2 n \cos \de{\frac{\pi}{n}}}{1+\cos \de{\frac{\pi}{n}}},
\end{equation}
therefore, if $n$ is odd, the Tsirelson bound $\Omega_{\rm QM}$ is upperbounded by $2 \vartheta(Y_n)-n$. The following quantum state and observables saturates this bound \cite{liang11}: $\ket{\psi} =(1,0,0)$ and $X_j=2\ketbra{v_j}{v_j} -\id$, where $\ket{v_j}= \left(\cos\theta,\sin\theta\cos[j \pi (n-1)/n],\sin\theta\sin[j \pi (n-1)/n]\right)$ and $\cos^2\theta =\cos(\pi/n)/(1+\cos(\pi/n))$.

	For even $n$, the proof can be obtained simply by noting that our inequalities are closely related to the Braunstein-Caves inequalities \cite{braunstein89}, whose Tsirelson bound was found in \cite{wehner06}. A small modification of the proof in \cite{wehner06} then suffices. The following quantum state and observables saturates this bound: $\ket{\psi}=(0,1/\sqrt{2},-1/\sqrt{2},0)$ and $X_j = \tilde{X_j} \otimes \id$ for even $j$ and $X_j = \id \otimes \tilde{X_j}$ for odd $j$, where $\tilde{X_j} = \cos(j \pi/n) \sigma_x +\sin(j \pi/n) \sigma_z$ and $\sigma_x,\sigma_z$ are Pauli matrices. 

The calculations for $\vartheta(Y_n)$ and the proof for even $n$ are presented in Appendix \ref{app:b}.
\end{proof}

	It is also interesting to examine the even case with the same technique we used for the odd case. If $n$ is even, the graph $G$ associated to $\Sigma$ is the Möbius ladder of order $2n$, $M_{2n}$ (see Fig.~\ref{Fig1}). We conjecture its $\vartheta$-function to be
\begin{equation}\label{eq:thetam}
 \vartheta(M_{2n})= \frac{n}{2} \de{1+\cos\frac{\pi}{n}},
\end{equation}
for which we present evidence in Appendix \ref{app:c}.

It can also be proved that these choices of state and observables saturating the quantum bounds are optimal, in the sense that such bounds cannot be reached in a Hilbert space of lower dimension. In fact, for odd $n$ there is nothing left to prove, since according to our definition of contextuality there is no contextual behaviour in a $2$-dimensional Hilbert space. For even $n$ it can be proved (see Appendix \ref{app:d}) that in a $3$-dimensional Hilbert space
\begin{equation}
\Omega_{QM3D}^{n}=\Omega_{QM}^{n-1} + 1 \ .
\end{equation}

This fact can be used as a dimension witness \cite{brunner08b}.

%%%%%%%%%%%%%%%%%%%%%%%%%%%%%%%%%%%%%%%%%%%%%%%%%%%%%%%%%%%%%%%%%%%
% Fig. 1
%%%%%%%%%%%%%%%%%%%%%%%%%%%%%%%%%%%%%%%%%%%%%%%%%%%%%%%%%%%%%%%%%%%

\begin{figure}[t]
\centerline{\includegraphics[width=8.4cm]{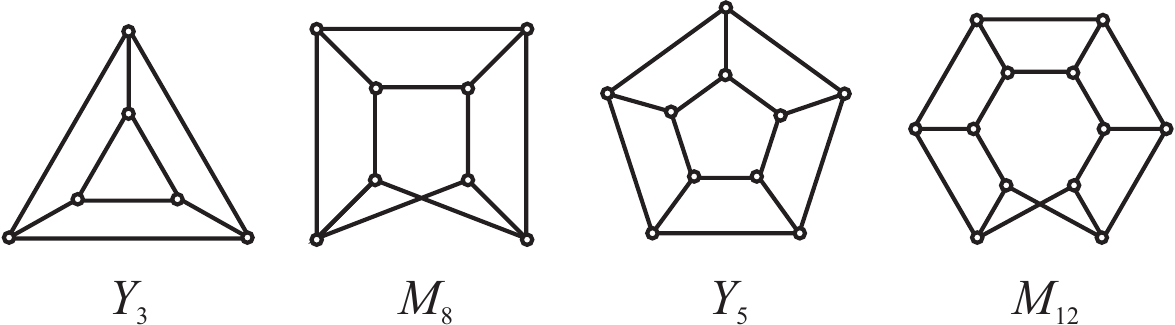}}
\caption{\label{Fig1}Graphs associated to the sum of probabilities $\Sigma$ in the tight noncontextuality inequalities for $n=3,4,5,6$.}
\end{figure}

%%%%%%%%%%%%%%%%%%%%%%%%%%%%%%%%%%%%%%%%%%%%%%%%%%%%%%%%%%%%%%%%%%%

\section{Observations}The quantum bounds for odd $n$ were found first in \cite{liang11,cabello10}, and for even $n$ in \cite{wehner06} in relation with Braunstein-Caves inequalities \cite{braunstein89}. However, we do think it is enlightening to show how graph theory provides a simple and unified approach to the problem.

Another observation is that while Braunstein-Caves inequalities are not tight Bell inequalities \cite{pitowsky01,collins04}, our inequalities \eqref{eq:boolenciclo} are tight noncontextuality inequalities. This is possible because the locality and contextuality scenarios are different: In the case of Bell inequalities, we demand every $X_i$ with even $i$ to be measurable together with every $X_j$ with odd $j$, and so the graph that represent these relations is the complete bipartite graph $K_{n/2,n/2}$, which is not isomorphic to the $n$-cycle (except for $n=4$, the CHSH case).

%%%%%%%%%%%%%%%%%%%%%%%%%%%%%%%%%%%%%%%%%%%%%%%%%%%%%%%%%%%%%%%%%%%

\section{Conclusions}The $n$-cycle contextuality scenario is the natural generalization of CHSH \citep{chsh69,fine82} and KCBS \cite{klyachko08} scenarios, the most fundamental scenarios for locality and noncontextuality, and has recently attracted increasing attention \citep{cabello10,liang11,fritz11,chaves12,chaves13}.
We have provided the complete characterization of the associated set of classical correlations for an arbitrary number $n$ of settings, the only other example of this kind being the Bell bipartite scenario with two observables for Alice and $n$ for Bob \citep{sliwa03,collins04}.
We have explicitly obtained the maximum quantum violation of all these inequalities with the minimal quantum dimension. We also completely characterized the associated no-disturbance correlations by finding the vertices of the corresponding polytope. 

%%%%%%%%%%%%%%%%%%%%%%%%%%%%%%%%%%%%%%%%%%%%%%%%%%%%%%%%%%%%%%%%%%%

\begin{acknowledgments}
The authors thank A. Acín, N. Brunner, P. Kurzyński, J. R. Portillo, R. Rabelo, S. Severini, and T. Vértesi for discussions. This work was supported by the Brazilian program Science without Borders, the Brazilian agencies Fapemig, Capes, CNPq, and INCT-IQ, the Spanish Project No.~FIS2011-29400, and EU (Marie-Curie CIG 293933/ENFOQI), the Austrian Science Fund (FWF): Y376-N16 (START prize), the BMBF (CHIST-ERA network QUASAR), the John Templeton Foundation, and the Swiss National Science Foundation.
\end{acknowledgments}

\appendix

\section{Detailed proof of Theorem 1}\label{app:a}

	Here we present the details missing in the proof of Theorem \ref{th:main}, namely the proof that the existence of a classical probability model for the observables  $\{X_0,\ldots,X_{n-1}\}$ is equivalent to the existence of classical models for $\{X_0,\ldots,X_{n-2}\}$ and $\{X_0,X_{n-1},X_{n-2}\}$, coinciding on their intersection $\{X_0,X_{n-2}\}$. 

	The first step is to extend the definition of graph representation given in Fig.\ref{Fig0} and \ref{Fig2} (graphs in Fig.\ref{Fig1} have a different interpretation).  We said that nodes represent dichotomic observables and edges represent compatibility relations, which means that if two nodes are connected by an edge the corresponding pair of observables admit a classical probability model. Such a definition can be generalized as follows
\begin{itemize}
\item[$(i)$] a node represents a subset of observables,
\item[$(ii)$] if two nodes are connected by an edge, then the corresponding subset of observables, \ie the union of the two subsets, admits a classical probability model. 
\end{itemize}

Such classical models are, in general, extension to broader subsets of the classical probability models associated  to subsets of commuting observables by QM (\ie spectral theorem). We can now recall the following result \cite{budroni10}:
\begin{lemma}\label{lemma1}
A set of probability assignments associated with a tree graph always admits a classical probability model.
\end{lemma}

Our strategy for the proof is then depicted in Fig.\ref{Fig4}: If the two subsets of observables in Fig.\ref{Fig4} (b) $(n-1)$-cycle and $3$-cycle, admit a classical representation, \ie all the corresponding inequalities are satisfied, then the set of probabilities can be extended, following $(i),(ii)$, as in Fig.\ref{Fig4} (c), \ie two classical model for $\{0,1,\ldots,n-3,n-2\}$ and for $\{0,n-1,n-2\}$  coinciding on their intersection $\{0,n-2\}$. By Lemma \ref{lemma1}, such a set already admits a classical representation.

\begin{figure}[t]
	\centering
	\includegraphics[width=1.10\columnwidth]{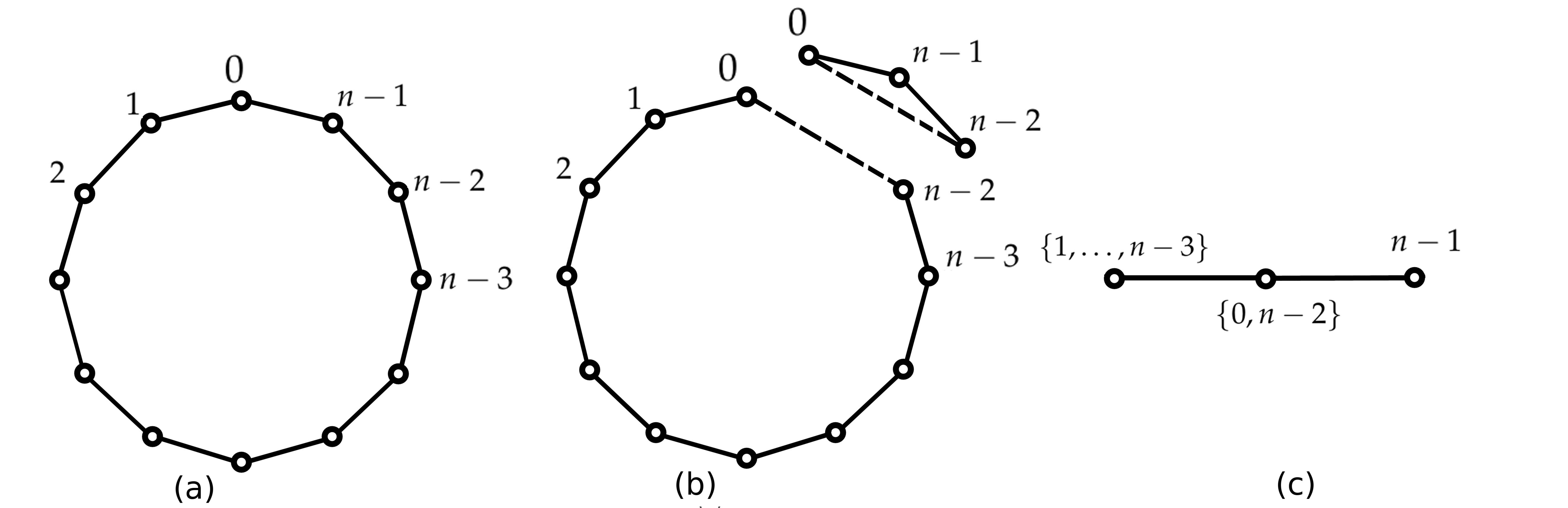}
	\caption{\label{Fig4} (a) $n$-cycle scenario. (b) subsets of observables that can be associated with the $(n-1)$-cycle and $3$-cycle scenario by considering the ``unmeasurable correlation'' $\mean{X_0 X_{n-2}}$ (dashed line). (c) Extended classical model that can be obtained if the two subset admits a classical representation coinciding on their intersection. Such a model is automatically classical as it can be depicted as a tree graph.}
\end{figure}

By the above procedure, we obtain a set of conditions that includes the ``unmeasurable correlation'' $\mean{X_0X_{n-2}}$ which plays a fundamental role since it constrains the two models on their intersection, but it is not actually measurable in the $n$-cycle scenario (see Fig.\ref{Fig4} (a)).
Such a variable must be therefore eliminated from the final set of conditions by means of Fourier-Motzkin (FM) elimination \cite{budroni12a}. We recall that FM elimination of a variable from a system of linear inequalities consists in summing each pair of inequalities where such a variable appears, respectively, with $+1$ and $-1$ coefficient (after a proper normalization of the inequalities) and keeping the inequalities where it does not appear \cite{ziegler95}. As a result, the final system of inequalities admits a solution if and only if the initial system of inequalities does.

To summarize: if a set of probability assignment for the $n$-cycle scenario satisfies the set of inequalities obtained as FM elimination of the variable $\mean{X_0 X_{n-2}}$ from the set of inequalities for $(n-1)$-cycle (for $\{0,\ldots,n-2\}$) and the $3$-cycle (for $\{0,n-2,n-1\}$), then both subsets of observables admits a classical representation with consistent assignments for $\mean{X_0 X_{n-2}}$, \ie such representations coincide on their intersection, as depict in Fig.\ref{Fig4} (c). By Lemma \ref{lemma1} these conditions are already sufficient for the existence of a classical model for the whole set of observables $\{0,1,\ldots,n-2,n-1\}$.

\section{Proof of tightness of the inequalities}\label{app:tightness}

Tightness can be proved by showing that inequalities (\ref{eq:boolenciclo}) correspond to facets of the $2n$-dimensional correlation polytope, \ie, they are saturated by $2n$ noncontextual vertices which generate an affine subspace of dimension $2n-1$. First, focus on the inequality of the odd $n$-cycle for which all $\gamma_i = -1$. It is saturated by $2n$ vertices which can be written as $(\pm v_i, w_i)$, for $i=0,\ldots,n-1$, where $w_i$ is a $n$-dimensional vector given by a cyclic permutation of the components of $w_{0}=(+1,-1,-1,\ldots,-1)$ and $v_i$ is the vector with $i$th component equal to $+1$ that satisfies relation \eqref{eq:nc}. Then it holds that $v_i+v_{i+1}=2 e_{i+1}$, where $\{e_0,\ldots,e_{n-1}\}$ is the canonical basis of $\mathbb{R}^n$, and $w_i+(1,1,\ldots,1)=2 e_i$. As a consequence, $\{ (\pm v_i,w_i)\}_{i=1,\ldots,2n}$ is a basis for $\mathbb{R}^{2n}$, showing independence. Since all the other vertices and inequalities are obtained from this one via the mapping $X_i \mapsto -X_i$, this proves the odd $n$ case. The proof for even $n$ is analogous.

\section{Detailed proof of theorem 3}\label{app:b}

An orthonormal representation (OR) for a graph ${G=(V,E)}$ is a set of unit vectors $\{ v_i \}$ associated with vertices $V=\{i\}$ such that two vectors are orthogonal if the corresponding vertices are adjacent, \ie $(i,j)\in E$. Lov\'{a}sz $\vartheta$ function is defined as the maximum, over all OR, of the norm of the operator given by sum of the unidimensional projectors associated with vectors \cite{lovasz79,lovasz09}. Notice that different vertices can be mapped onto the same vector, but then the corresponding projector appears in the sum once for each vertex associated with it.

For the prism graph $Y_n$, in general, it holds ${\vartheta(Y_n) \leq 2 \vartheta(C_n)=\frac{2 n \cos \left(\frac{\pi}{n}\right)}{1+\cos \left(\frac{\pi}{n}\right)}}$ since a graph consisting in two copies of $C_n$, let us denote it as $G$, can be obtained from $Y_n$ by removing the edges connecting vertices of the outer cycle with those of the inner cycle.

 Consider an OR for $C_n$, say $v_0,\ldots,v_{n-1}$, which gives the maximum value for the norm of the corresponding sum of projectors, i.e $\vartheta(C_n)$. Clearly, the $2n$ vectors $v_i, v'_i$, with $v'_i=v_i$, for $i=0,\ldots,n-1$, form a OR for $G$, giving $\vartheta(G)=2\vartheta(C_n)$. To show that $\vartheta(Y_n)=\vartheta(G)=2\vartheta(C_n)$, it is sufficient to notice that the above vectors are also an OR for $Y_n$. Such an OR is obtained by associating $v_i$ with the $i$th vertex of the outer cycle and the vector $v'_{i+1}$ with the $i$th vertex of the inner cycle. This completes the discussion for the case of odd $n$.

For the case of even $n$, the proof is based on positive semidefiniteness conditions analogous to those discussed in \cite{wehner06, navascues07, navascues08}. Via them we can show that Eq.~\ref{eq:thetam} is an upper bound to the Tsirelson bound, and the proof is completed by noting that we already provided quantum observables and states saturating it.

Let us consider a quantum state $\rho$ and $n$ dichotomic observables $X_0,\ldots,X_{n-1}$ with even $n$. Then the complex matrix $\Gamma_{ij}=\tr (\rho X_i X_j)$ must be positive semidefinite. In fact, given a complex vector $v$, we have
\begin{equation}\label{e:exp}
\begin{split}
v^\dagger \Gamma v &= \sum_{ij} v_i^* \Gamma_{ij} v_j = \tr(\rho\ \sum_{ij} v_i^* v_j X_i  X_j)\\
&=\tr(\rho\ \sum_{i} v_i^*X_i \sum_j v_j  X_j) = \tr(\rho O^\dagger O ) \geq 0,
\end{split}
\end{equation}
with $O \equiv \sum_i v_i  X_i $. An upper bound for the quantum violation of the expression
\begin{equation}
\sum_{i=0}^{n-1}\gamma_i \mean{X_i X_{i+1}},
\end{equation}
with $\gamma_{n-1}=-1$ and all other coefficients $+1$, can be therefore obtained as the semidefinite program (SDP)
\begin{equation}\label{e:sdp}\begin{split}
 \text{maximize:}&\quad \frac{1}{2} \tr(\beta\ \Gamma),\\
 \text{subject to:}&\quad \Gamma \succeq 0, \quad \Gamma_{ii}=1 ,
\end{split}\end{equation}
where $\beta$ is a symmetric real matrix such that ${\frac{1}{2} \tr(\beta\ \Gamma)=  \sum_{i=0}^{n-1}\gamma_i \Gamma_{i,i+1}}$. The optimality of the solution $n\cos(\frac{\pi}{n})$
for the above SDP, up to a reordering of the coordinates, has been proved by Wehner \cite{wehner06}. Together with the explicit state and observables presented in the main text, this concludes our proof.

\section{\texorpdfstring{Evidence for the conjectured Lovász $\vartheta$ function for Möbius ladder graphs}{Evidence for the conjectured Lovász function for Möbius ladder graphs}}\label{app:c}

	In Eq.~\ref{eq:thetam} we conjectured an expression for $\vartheta(M_{2n})$. The evidence we have for it is both numerical and mathematical: We calculated explicitly the value for $\vartheta(M_{2n})$ for even $n$  up to $n=64$, \ie $\vartheta(M_{128})$ and it coincides with the expression given in Eq.~\ref{eq:thetam} with very high precision. Moreover, since $M_{2n}$ is a regular graph (each vertex has the same number of neighbours) $\vartheta(M_{2n})$ can be upperbounded by the expression \cite{lovasz79}
\begin{equation}
\vartheta(M_{2n})\leq \frac{-n\lambda_{2n}}{\lambda_1 - \lambda_{2n}},
\end{equation}
where $\lambda_1\geq \lambda_2\geq \ldots \geq \lambda_{2n}$ are the eigenvalues for the adjacency matrix $A$ for $M_{2n}$. Since $A$ is a circulant matrix they can be explicitly computed \cite{gray06}, and with them we obtain
\begin{equation}\label{reg}
\vartheta(M_{2n})\leq \frac{n (2\cos(\frac{\pi}{n}) +1)}{2+ \cos(\frac{\pi}{n})}.
\end{equation}
Comparing (\ref{reg}) with our conjecture for $\vartheta(M_{2n})$ in the asymptotic limit $n\rightarrow \infty$, we obtain
\begin{equation}
\frac{n (2\cos(\frac{\pi}{n}) +1)}{2+ \cos(\frac{\pi}{n})} - \frac{n}{2} \de{1+\cos\frac{\pi}{n}} \approx \frac{\pi^2}{12n}.
\end{equation}

\section{\texorpdfstring{Quantum bounds for even $n$ in dimension $3$}{Quantum bounds for even n in dimension 3}}\label{app:d}
Consider the inequalities 
\begin{equation}
 \label{eq:boolencicloapp}
 \Omega=\sum_{i=0}^{n-1} \gamma_i \langle X_{i} X_{i+1} \rangle
\stackrel{\text{\tiny{\textup{NCHV}}}}{\leq} n - 2,
\end{equation}
where  $\gamma_i \in \{-1,1\}$ such that the 
number of $\gamma_i = -1$ is odd, and $n$ is even. 
By the symmetry of the problem, namely the fact that each inequality is 
obtained via the substitution  $X_i \rightarrow -X_i$ for some indices $i$, the quantum bound 
for (\ref{eq:boolencicloapp}) must be the same for all possible choices of $\gamma$.

Let us start with $n$ general 3-dimensional observables $X_i$ and a vector $\gamma$ giving the \textit{lhs} of Eq.~\eqref{eq:boolencicloapp}. Since we are in 3 dimensions, for each $i$, either $X_i$ or $-X_i$, is 
given by a 1-dimensional projector $P_i$, as $\pm X_i= 2P_i -1$. The substitution $X_i 
\rightarrow - X_i$ simply amounts to a new definition of the vector $\gamma$. We have therefore a 
new expression (\ref{eq:boolencicloapp}) where all the $X_i$'s are given by $1$-dimensional 
projectors.
For such observables it holds
\begin{equation}\label{eq:comm}
[X_i,X_{i+1}] = 0 \Longleftrightarrow P_i P_{i+1}=0 \text{ or } P_i=P_{i+1}.
\end{equation}
Let us assume, for the moment, that the condition $P_i P_{i+1}=0$ holds for all $i=0,\ldots,n-1$, we shall discuss the other cases later. We want to calculate the maximum of the \textit{lhs} of (\ref{eq:boolencicloapp}) over all $\gamma$, namely
\begin{equation}
 \label{eq:buni}
 \max_{\gamma, P_i,\rho} \sum_{i=0}^{n-1} \gamma_i \langle (2P_i-1)(2P_{i+1}-1) \rangle,
\end{equation}
which can be rewritten as
\begin{equation}
\begin{split}
\max_{\gamma, P_i,\rho} \sum_{i=0}^{n-1} \gamma_i [1- 2 \langle P_i + P_{i+1} \rangle]=\\
=\max_{\gamma, P_i,\rho} \frac{1}{2 }\sum_{i=0}^{n-1} [-(\gamma_i + \gamma_{i-1})]\ (4\langle P_i \rangle -1 ). \\
\end{split}
\end{equation}
Since the number of $\gamma_i=-1$ must be odd, and $n$ is even, at least two terms $(\gamma_i + \gamma_{i-1})$ and $(\gamma_{i+1} + \gamma_{i})$ must be zero. Without loss of generality, we can assume it holds for $i=n-2$. We have therefore
\begin{equation}
\begin{split}
\max_{\gamma, P_i,\rho} \frac{1}{2 }\sum_{i=0}^{n-1} [-(\gamma_i + \gamma_{i-1})]\ (4\langle P_i \rangle -1 )\leq \max_{ P_i,\rho}\  4\sum_{i=0}^{n-3} \langle P_i \rangle\\ -(n-2)= 2(n-2)-(n-2)=n-2,
\end{split}
\end{equation}
where we used that the maximum of $\sum_{i=0}^{n-3} \langle P_i \rangle$ is bounded by 
 $\frac{n-2}{2}$. In fact, $\langle P_i + P_{i+1}\rangle \leq 1$, since their sum is still a projector.

We must now consider the other possibilities given by (\ref{eq:comm}). If for a given index, say $i=0$, $P_i=P_{i+1}$, we simply have that $X_0= X_{1}$, therefore Eq.~\eqref{eq:boolencicloapp} reduces to 
\begin{equation}
\begin{split}
 \label{eq:n-1}
 \sum_{i=0}^{n-1} \gamma_i \langle X_{i} X_{i+1} \rangle = \gamma_0 + \sum_{i=1}^{n-1} \gamma_i \langle X_{i} X_{i+1}\rangle
 \leq \Omega_{QM}^{n-1} + 1 ,
\end{split}
\end{equation}
since $\langle X_0 X_{n-1}\rangle= \langle X_1 X_{n-1}\rangle$.

If for two indices, $i,j$, $P_i=P_{i+1}$ and $P_j=P_{j+1}$, the problem is reduced to the case $n-2$, and so on for all the other cases.

We have therefore proved that the \textit{optimal bound} for $n$-cycle inequalities in 3 dimension is given by 
\begin{equation}
\Omega_{QM3D}^{n}=\Omega_{QM}^{n-1} + 1 \text{ , for  even } n.
\end{equation}
Remember that $\Omega_{QM3D}^{n} \geq n-2$ since $\Omega_{QM}^{n-1} \geq n-3$ and that the bound in (\ref{eq:n-1}) can always be achieved with $1$-dimensional projectors.

%%%%%%%%%%%%%%%%%%%%%%%%%%%%%%%%%%%%%%%%%%%%%%%%%%%%%%%%%%%%%%%%%%%

\bibliographystyle{linksen}

%%%%%%%%%%%%%%%%%%%%%%%%%%%%%%%%%%%%%%%%%%%%%%%%%%%%%%%%%%%%%%%%%%%

\bibliography{biblio}
\end{document}